\documentclass[pra,twocolumn,floatfix,aps,superscriptaddress,showpacs,amsfonts]{revtex4}
\usepackage{graphicx}
\usepackage{array}
\usepackage{amsthm}
\usepackage{amsmath}
\usepackage{amssymb}
\newtheorem{theorem}{Theorem}

\begin{document}
\title{The Fidelity and Trace Norm Distances for Quantifying Coherence}

\author{Lian-He Shao}
\email{snnulhs@gmail.com}
\affiliation{College of Computer Science, Shaanxi Normal University, Xi'an, 710062,
P. R. China}
\affiliation{Beijing National Laboratory for Condensed Matter Physics,
Institute of Physics, Chinese Academy of Sciences, Beijing, 100190, P. R. China}

\author{Zhengjun Xi}
\email{xizhengjun@snnu.edu.cn}
\affiliation{College of Computer Science, Shaanxi Normal University, Xi'an, 710062,
P. R. China}

\author{Heng Fan}
\email{hfan@iphy.ac.cn}
\affiliation{Beijing National Laboratory for Condensed Matter Physics,
Institute of Physics, Chinese Academy of Sciences, Beijing, 100190, P. R. China}
\affiliation{Collaborative Innovation Center of Quantum Matter, Beijing, P. R. China}

\author{Yongming Li}
\email{liyongm@snnu.edu.cn}
\affiliation{College of Computer Science, Shaanxi Normal University, Xi'an, 710062,
P. R. China}

\begin{abstract}

We investigate the coherence measures induced by fidelity and trace norm,
based on the recent proposed coherence quantification in [Phys. Rev. Lett. 113, 140401, 2014].
We show that the fidelity of coherence does not in general satisfy the monotonicity requirement
as a measure of coherence under the subselection of measurements condition.
We find that the trace norm of coherence can act as a measure of coherence for qubit case
and some special class of qutrits. 
\end{abstract}
\eid{identifier}
\pacs{}
\maketitle

\section{Introduction}
Coherence arising from quantum superposition which plays a central role for 
quantum mechanics. Quantum coherence is an important subject in quantum theory
and quantum information science which is a common necessary condition for 
both entanglement and other types of quantum correlations. 
It has been shown that a good definition of coherence does not only depend on the state of the system $\rho$, 
but also depends on the a fixed basis for the quantum system~\cite{Baumgratz13}.
Up to now, several themes of coherence have been considered such as 
witnessing coherence~\cite{CMLi12}, catalytic coherence~\cite{Aberg13}, the thermodynamics of quantum coherence~\cite{Rosario13}, and the role of coherence in biological system~\cite{Huelga13}. 
There seems no well-accepted efficient method for quantifying coherence until recently. 
Girolami proposed a measure of quantum coherence based on the Wigner-Yanase-Dyson skew information~\cite{Girolami14}. It is not only in theoretical but also an experimental scheme implementable with current technology. Baumgratz $et$ $al$. introduced a rigorous framework for quantification of coherence and proposed several measures of coherence, which are based on the well-behaved metrics including the $l_p$-norm, relative entropy, trace norm and fidelity~\cite{Baumgratz13}.
The quantification of coherence promoted in a unified 
and rigorous framework thus stimulated a lot of further considerations about quantum coherence~\cite{ZJXI14,Marvian14,Monras13,Rivas14}.

From the view point of the definition, one can straightforwardly quantify the coherence in a given basis by measuring the distance between the quantum state $\rho$ and its nearest incoherent state. 
This property is similar as that of the well studied measures of the quantum correlation, e.g., entanglement and quantum discord~\cite{Vedral97,Modi10,Vedral01}. We remark that the coherence measures are to be applied to one quantum system
but quantum correlation measures naturally involve more than two parties.
We know that several basic criteria are proposed which should 
be satisfied by any measure of the entanglement \cite{Vedral97,Vedral01}. 
In comparison, the coherence measures also need to satisfy the following four necessary criteria as presented in Ref.\cite{Baumgratz13}.
Given a finite-dimensional Hilbert space $\mathcal{H}$ with $d=dim(\mathcal{H})$. We note that $\mathcal{I}$ is the set of quantum states which is called incoherent state that are diagonal in a fixed basis $\{|i\rangle\}_{i=1}^d$. Then any proper measure of the coherence $C$ must satisfy the following conditions:

(C1) $C(\delta)=0$ for all $\delta\in \mathcal{I}$.

(C2a) Monotonicity under all the incoherent completely positive and trace preserving (ICPTP) maps $\Phi$: $C(\rho)\geq C(\Phi(\rho)) $.

(C2b) Monotonicity for average coherence under subselection based on measurements outcomes: $C(\rho)\geq \sum_n p_n C(\rho_n) $ for all $\{K_n\}$ with $\sum_n K_n^{\dagger}K_n= \mathbb{I}$ and $K_n \mathcal{I} K_n^\dagger\subset \mathcal{I}$.

(C3) Non-increasing under mixing of quantum states: $\sum_n p_n C(\rho_n)\geq C(\sum_n p_n \rho_n )$ for any set of states $\{\rho_n\}$ and any
$p_n\geq0$ with $\sum_n p_n=1$.

As shown in \cite{Baumgratz13}, the condition (C2b) is important as it allows for sub-selection
based on measurement outcomes, a process available in well controlled quantum experiments. It has been shown that the quantum relative entropy and $l_1$-norm satisfy this condition. The squared Hilbert-Schmidt norm does not satisfy (C2b). 
However, it is still an open question whether some other coherence measures satisfy (C2b). 
In this paper, we will show that the measure of coherence induced by fidelity defined distance 
does not satisfy condition (C2b). Explicit example is presented. 
We will also show that trace norm of coherence for qubit satisfies condition (C2b),
the case of qutrit, which is in three-dimensional Hilbert space, is in general unknown, but for some special qutrits, trace norm of coherence
satisfies this condition.

This paper is organized as follows. In Sec.~\ref{sec:FOC}, 
we illustrate that the fidelity of coherence is not a good measure for quantum coherence by presenting 
an example that condition (C2b) is not satisfied. 
In Sec.~\ref{sec:TNC}, we show that condition (C2b) can be satisfied in qubit case and
some special qutrits for trace norm of coherence. We summarize our results in Sec.~\ref{sec:conclusion}.

\section{Fidelity of Coherence}\label{sec:FOC}
As a measure of distance, the fidelity~\cite{Jozsa94} $ F(\rho,\delta)=[tr\sqrt{\rho^{\frac{1}{2}}\delta\rho^{\frac{1}{2}}}]^2 $ is non-decreasing under CPTP maps $\varepsilon$, e.g., $F(\varepsilon(\rho),\varepsilon(\delta))\geq F(\rho,\delta) $. 
Then we know that the fidelity induced distance $1-\sqrt{F(\rho,\delta)}$ is monotonicity under ICPTP maps, and $F(\rho,\delta)=1$ iff $\rho=\delta$. Hence, the fidelity of coherence can be defined as:
\begin{equation}\label{Eq:rho coh}
C_F(\rho)=\min_{\delta\in\mathcal{I}}D(\rho,\delta)=1-\sqrt{\max_{\delta\in\mathcal{I}}F(\rho,\delta)}
\end{equation}
It is easy to find that the fidelity of coherence fulfils (C1),(C2a) and (C3)~\cite{Baumgratz13}.

For the condition (C2b), without loss of generality, we consider the one-qubit system. It is known that for a qubit, the fidelity has a simple form. From the Bloch sphere representation of a quantum state, $\rho$ and $\delta$ can be expressed as \cite{Nielsen},
\begin{equation}\label{Eq:rho del}
\rho=\frac{\mathbb{I}+\mathbf{r}\cdot\mathbf{\sigma}}{2},\delta=\frac{\mathbb{I}+\mathbf{s}\cdot\mathbf{\sigma}}{2}
\end{equation}
where $\mathbb{I}$ is the identity operator, $\mathbf{r}=(r_x,r_y,r_z)$ and $\mathbf{s}=(s_x,s_y,s_z)$ are the Bloch vectors and $\mathbf{\sigma}=(\sigma_x,\sigma_y,\sigma_z)$ is a vector of Pauli matrices. Then the fidelity for qubits has an elegant form, 
\begin{equation}\label{Eq:FIDE}
F(\rho,\delta)=\frac{1}{2}\left[1+\mathbf{r}\cdot\mathbf{s}+\sqrt{(1-|\mathbf{r}|^2)(1-|\mathbf{s}|^2)}\right],
\end{equation}
where $\mathbf{r}\cdot\mathbf{s}$ is the inner product of $\mathbf{r}$ and $\mathbf{s}$, $|\mathbf{r}|$ and $|\mathbf{s}|$ is the magnitude of $\mathbf{r}$ and $\mathbf{s}$, respectively.

Because $\delta$ is the incoherent state, then the Bloch vector $\mathbf{s}$ can be expressed as $\mathbf{s}=(0,0,s_z)$, the Eq.~\eqref{Eq:FIDE} can be replaced as,
\begin{equation}\label{Eq:REFIDE}
F(\rho,\delta)=\frac{1}{2}\left[1+r_zs_z+\sqrt{(1-r_x^2-r_y^2-r_z^2)(1-s_z^2)}\right].
\end{equation}
In order to obtain $\underset{\delta\in\mathcal{I}}\max F(\rho,\delta)$, we should take derivative with respect to $s_z$, then we have,
\begin{equation}\label{Eq:DFIDE}
\frac{dF(\rho,\delta)}{ds_z}=\frac{1}{2}\left[r_z-\sqrt{(1-r_x^2-r_y^2-r_z^2)}\frac{s_z}{\sqrt{1-s_z^2}}\right].
\end{equation}
After some simple algebraic operation, we can obtain,
\begin{equation}\label{Eq:MAXFIDE}
\max_{\delta\in\mathcal{I}}F(\rho,\delta)=\frac{1}{2}\left[1+\sqrt{(1-r_x^2-r_y^2)}\right].
\end{equation}
Therefore, we obtain
\begin{eqnarray}
C_F(\rho)&=&1-\sqrt{\max_{\delta\in\mathcal{I}}F(\rho,\delta)} \nonumber \\
&=&1-\frac{\sqrt{2}}{2}\sqrt{1+\sqrt{(1-r_x^2-r_y^2)}}. \label{Eq:COFIDE}
\end{eqnarray}
This implies that the state $\rho_{\mathrm{diag}}$ is not necessarily optimized for the fidelity of coherence in the one-qubit system. Thus, in general, we have \begin{equation}
\min_{\delta\in\mathcal{I}}(1-\sqrt{F(\rho,\delta)})\neq 1-\sqrt{F(\rho,\rho_{\mathrm{diag}})}.
\end{equation}
This makes that the sub-selection process becomes hard to verify. We should choose peculiar incoherent operations to simplify calculation.

Now we give an example to show that the condition (C2b) is violated. As we know that the depolarizing, the phase-damping, and the amplitude-damping channels are the qubit incoherent operatorations. We choose the amplitude-damping-like operation as incoherent operations, its operation elements are expressed as,
\begin{equation}K_1=\left(
  \begin{array}{lcr}
 a & 0 \\
 0 & b
  \end{array}
\right), K_2=\left(
  \begin{array}{lcr}
 0 & c \\
 0 & 0
  \end{array}
  \right).
\end{equation} After applying it on the one-qubit, we obtained the output state
\begin{eqnarray}\label{Eq:KR}
\rho_1
 =\left( \begin{array}{lcr}
\frac{|a|^2(1+r_z)}{|a|^2(1+r_z)+|b|^2(1-r_z)} & \frac{ab^*(r_x-ir_y)}{|a|^2(1+r_z)+|b|^2(1-r_z)} \\
  \frac{a^*b(r_x+ir_y)}{|a|^2(1+r_z)+|b|^2(1-r_z)} & \frac{|b|^2(1-r_z)}{|a|^2(1+r_z)+|b|^2(1-r_z)}
  \end{array}
\right),
\end{eqnarray}
with the probability
\begin{equation}\label{Eq:PRO}
p_1=tr(K_1\rho K_1^\dagger)=\frac{1}{2}\left(|a|^2(1+r_z)+|b|^2(1-r_z)\right).
\end{equation}

In order to obtain the quantity, $C_F(\rho_1)$, 
we should transform $\rho_1$ to the Bloch representation. Then the Bloch vector for $\rho_1$ can be given by,
\begin{equation}
\begin{cases}\label{Eq:BLOCH}
    x=\frac{ab^*(r_x-ir_y)+a^*b(r_x+ir_y)}{|a|^2(1+r_z)+|b|^2(1-r_z)}\\
    y=\frac{-i[a^*b(r_x+ir_y)-ab^*(r_x-ir_y)]}{|a|^2(1+r_z)+|b|^2(1-r_z)}\\
    z=\frac{|a|^2(1+r_z)-|b|^2(1-r_z)}{|a|^2(1+r_z)+|b|^2(1-r_z)}
  \end{cases}
\end{equation}
The fidelity of coherence for $\rho_1$ is obtained by substituting $x$ and $y$ as,
\begin{eqnarray}
&&C_F(\rho_1)=1-\frac{\sqrt{2}}{2}\sqrt{1+\sqrt{(1-x^2-y^2)}} \nonumber \\
&&={\small 1-\frac{\sqrt{2}}{2}\sqrt{1+\sqrt{1-\frac{4|a|^2|b|^2(r_x^2+r_y^2)}{[|a|^2(1+r_z)+|b|^2(1-r_z)]^2}}}} \nonumber \\
\label{Eq:CF}
\end{eqnarray}
Because that $K_i$ should satisfy $\sum_n K_n^{\dagger}K_n= \mathbb{I}$, then we have $|a|^2=1,|b|^2+|c|^2=1$. Let $|b|^2=\frac{1}{4}, |c|^2=\frac{3}{4}, r_x^2+r_y^2=\frac{1}{2}, r_z^2=\frac{1}{2}$, substitute those values into Eq.~\eqref{Eq:PRO}, Eq.~\eqref{Eq:CF} and Eq.~\eqref{Eq:COFIDE}, then we have,
\begin{eqnarray}\label{Eq:PCFIDE}
p_1C(\rho_1)&=&\frac{10-3\sqrt{2}}{16}[1-\frac{\sqrt{2}}{2}\sqrt{1+\sqrt{1-\frac{32}{(10-3\sqrt{2})^2}}}] \nonumber \\
& \approx &0.08273
\end{eqnarray}
and
\begin{eqnarray}
C_F(\rho)=1-\sqrt{\frac{1}{2}[1+\frac{\sqrt{2}}{2}]}\approx0.07612.
\end{eqnarray}
Note that the operation $K_2$ makes $C_F(\rho_2)=0$. Thus, we obtain \begin{equation}
\sum_{i=1}^2 p_iC_F(\rho_i)=p_1C_F(\rho_1)>C_F(\rho).
\end{equation}

From the above example, we then conclude that the condition (C2b): $C_F(\rho)\geq\sum_n p_n C_F(\rho_n)$ is 
not generally true for ICTPT maps for measure of coherence induced
by fidelity. If the Bloch vector $\mathbf{r}=(r_x,r_y,r_z)$ satisfies $r_x^2+r_y^2+r_z^2\leq1$, and suppose $-\frac{\sqrt{2}}{2}\leq r_z\leq \frac{\sqrt{2}}{2}$, then we can find many examples to illustrate that the fidelity of coherence does not satisfy condition (C2b), 
as shown in Fig.~\ref{Fig_1}.
\begin{figure}\label{fig1}
  \includegraphics[scale=0.3]{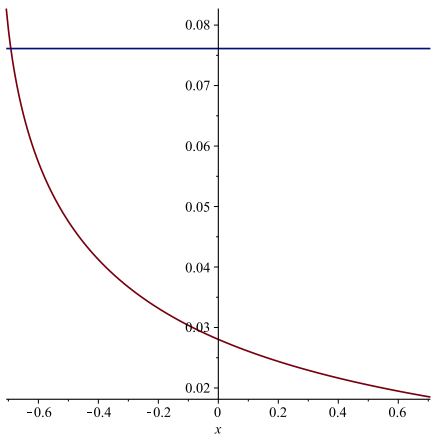}
\caption{ The Blue line shows $C_F(\rho)$, the red line shows $p_1C_F(\rho_1)$.
The x-axis expresses the values of $r_z$ ($-\frac{\sqrt{2}}{2} \leq r_z\leq \frac{\sqrt{2}}{2}$). The intersecting coordinate is $(-0.691964,0.076120)$. Because of $-0.691964>-\frac{\sqrt{2}}{2}$, so when $-0.691964\geq r_z\geq -\frac{\sqrt{2}}{2}$, we always have $p_1C(\rho_1)+p_2C(\rho_2)\geq C(\rho)$}.
\label{Fig_1}
\end{figure}

\section{Trace norm of Coherence}\label{sec:TNC}

For the trace norm of coherence, we will list some basic calculations, so that they can help us to judge whether the trace norm of coherence satisfies the condition (C2b).

At first, we have already considered the one-qubit states for the trace norm of coherence. Given two one-qubit states $\rho=\frac{\mathbb{I}+\mathbf{r}\cdot\mathbf{\sigma}}{2}$ and $\delta=\frac{\mathbb{I}+\mathbf{s}\cdot\mathbf{\sigma}}{2}$. The trace norm between $\rho$ and $\delta$ can be expressed as
\begin{equation}
D_{tr}(\rho,\delta)=|\mathbf{r}-\mathbf{s}|
\end{equation}
Then the trace norm of coherence can be easily expressed as:
\begin{eqnarray}
C_{tr}(\rho)&=&\min_{\delta \in \mathcal{I}} D_{Tr}(\rho,\delta) \nonumber \\
&=&\min_{\delta\in\mathcal{I}}\sqrt{(r_x-s_x)^2+(r_y-s_y)^2+(r_z-s_z)^2} \nonumber \\
\end{eqnarray}
For the incoherent states, we know that $s_x=s_y=0$, the trace norm of coherence can be simplified as
\begin{eqnarray}
C_{tr}(\rho)&=&\min_{\delta\in\mathcal{I}}\sqrt{r_x^2+r_y^2+(r_z-s_z)^2} \nonumber \\
&=&\| \rho-\rho_{\mathrm{diag}}\|_{tr}=\sqrt{r_x^2+r_y^2}
\end{eqnarray}
Note that $C_{tr}(\rho)$ has the same form of expression with the $l_1$ norm of coherence $C_{l_1}(\rho)=\sum_{i,j,i\neq j}|\rho_{i,j}|$ for the one-qubit case. So in this situation, $C_{tr}(\rho)$ satisfies the condition (C2b).
Here we simply conclude that trace norm can act as a coherence measure for a qubit. 

For the one-qutrit quantum system, the eigenvalues of the qutrit density matrices have complex expressions. 
It seems difficult to estimate the optimal incoherent state. Fortunately, we can find some special density matrices 
whose optimal incoherent states can be obtained.
\begin{theorem}
For the following three classes of qutrit states
\begin{equation}\rho_X=\left(
  \begin{array}{lcr}
 a_{11} & 0 & a_{13} \\
 0 & a_{2,2} & 0 \\
 a_{13}^* & 0 & a_{33}
  \end{array}
\right),
\end{equation}
\begin{equation}
\rho_Y=\left(
  \begin{array}{lcr}
 a_{11} & a_{12} & 0 \\
 a_{12}^* & a_{22} & 0 \\
 0 & 0 & a_{33}
  \end{array}
\right),
\end{equation}
and
\begin{equation}
\rho_Z=\left(
  \begin{array}{lcr}
 a_{11} & 0 & 0 \\
 0 & a_{22} & a_{23} \\
 0 & a_{23}^* & a_{33}
  \end{array}
\right),
\end{equation}
the optimal incoherent state of the trace norm of coherence is of the form $\rho_{\mathrm{diag}}$.
\end{theorem}
\begin{proof} We only prove the case of state $\rho_X$, the states $\rho_Y$ and $\rho_Z$ are completely analogous.
Since all qutrit incoherent states have the form as
\begin{equation}
\delta=\left(
  \begin{array}{lcr}
 x & 0 & 0 \\
 0 & y & 0 \\
 0 & 0 & z
  \end{array}
\right),
\end{equation}
then we can easily obtain the eigenvalues for $\rho_X-\delta$, 
\begin{equation}
\begin{cases}\label{Eq:eigen}
    \lambda_1=a_{22}-y,\\
    \lambda_2=\frac{y-a_{22}}{2}-\frac{\sqrt{(2x+y-2a_{11}-a_{22})^2+4|a_{13}|^2}}{2},\\
    \lambda_3=\frac{y-a_{22}}{2}+\frac{\sqrt{(2x+y-2a_{11}-a_{22})^2+4|a_{13}|^2}}{2}.
  \end{cases}
\end{equation}
We know that $\rho_X-\delta$ is a normal matrix, its singular values are the modulus of the eigenvalues for $\rho_X-\delta$, then we have,
\begin{equation}\label{Eq:lambda}
||\rho_X-\delta||_{tr}=|\lambda_1|+|\lambda_2|+|\lambda_3|
\end{equation}
In order to minimize $||\rho_X-\delta||_{tr}$ over all the incoherent states, we should consider four cases as following.

Case 1. When $\frac{y-a_{22}}{2}\geq\frac{\sqrt{(2x+y-2a_{11}-a_{22})^2+4|a_{13}|^2}}{2}$ and $a_{22}\leq y$, we can simplify Eq.~\eqref{Eq:lambda} as
\begin{eqnarray}\label{Eq:case1}
||\rho_X-\delta||_{tr} & = & 2y-2a_{22} \nonumber \\
& \geq & 2\sqrt{(2x+y-2a_{11}-a_{22})^2+4|a_{13}|^2} \nonumber \\
& \geq & 2 \sqrt{|a_{13}|^2}\nonumber\\
&=&||\rho_X-\rho_{\mathrm{diag}}||_{tr}.
\end{eqnarray}

Case 2. When $\frac{y-a_{22}}{2}\leq\frac{\sqrt{(2x+y-2a_{11}-a_{22})^2+4|a_{13}|^2}}{2}$ and $a_{22}\leq y$, similar to case 1, we have
\begin{eqnarray}\label{Eq:case2}
& &||\rho_X-\delta||_{tr} \nonumber\\
 &= &y-a_{22}+\sqrt{(2x+y-2a_{11}-a_{22})^2+4|a_{13}|^2} \nonumber \\
& \geq & 2 \sqrt{|a_{13}|^2}\nonumber\\
&=&||\rho_X-\rho_{\mathrm{diag}}||_{tr}.
\end{eqnarray}

Case 3. When $\frac{y-a_{22}}{2}\leq\frac{\sqrt{(2x+y-2a_{11}-a_{22})^2+4|a_{13}|^2}}{2}$, $y\leq a_{22}$ and $\frac{y-a_{22}}{2}+\frac{\sqrt{(2x+y-2a_{11}-a_{22})^2+4|a_{13}|^2}}{2}\geq 0$, we have
\begin{eqnarray}\label{Eq:case3}
& &||\rho_X-\delta||_{tr} \nonumber\\
& = & a_{22}-y+\sqrt{(2x+y-2a_{11}-a_{22})^2+4|a_{13}|^2} \nonumber \\
&\geq & 2\sqrt{|a_{13}|^2}\nonumber\\
&=&||\rho_X-\rho_{\mathrm{diag}}||_{tr}.
\end{eqnarray}

Case 4. When $\frac{y-a_{22}}{2}\leq\frac{\sqrt{(2x+y-2a_{11}-a_{22})^2+4|a_{13}|^2}}{2}$,  $y\leq a_{22}$ and $\frac{y-a_{22}}{2}+\frac{\sqrt{(2x+y-2a_{11}-a_{22})^2+4|a_{13}|^2}}{2}\leq 0$, we have
\begin{eqnarray}\label{Eq:case4}
& &||\rho_X-\delta||_{tr} \nonumber\\
& \geq & 2\sqrt{(2x+y-2a_{11}-a_{22})^2+4|a_{13}|^2} \nonumber \\
& \geq & 2\sqrt{|a_{13}|^2}\nonumber\\
&=&||\rho_X-\rho_{\mathrm{diag}}||_{tr}.
\end{eqnarray}
Through the above analysis, we can obtain that the trace norm of coherence for $\rho_X$ has the optimal incoherent state $\rho_{\mathrm{diag}}$.
\end{proof}

According to the above theorem, we can also obtain an analytical expression of the trace norm of coherence for $\rho_X$ as,
\begin{equation}
C_{tr}(\rho_X)=D_{tr}(\rho_X,\rho_{diag})=2|a_{13}|.
\end{equation}
Note that $C_{tr}(\rho_X)$ also has the same form of expression with the $l_1$ norm of coherence $C_{l_1}(\rho)=\sum_{i,j,i\neq j}|\rho_{i,j}|$ for the $\rho_X$. Based on this fact and as shown in Ref.\cite{Baumgratz13}, we know that $C_{tr}(\rho_X)$ satisfies the condition (C2b). 
Similarly, we can verify the trace norm of coherence for $\rho_Y$ and $\rho_Z$ satisfies the condition (C2b).

\section{conclusion}\label{sec:conclusion}
In this paper, we show that the fidelity of coherence does not satisfy condition (C2b) by presenting an example. 
We then conclude that the measure of coherence induced by fidelity is not a good measure for quantifying coherence. 
For the trace norm of coherence, we have shown that the qubit states and some special qutrit states 
can satisfy condition (C2b). Our results show that the trace norm of coherence
is equivalent to $l_1$ norm of coherence for qubits and special qutrits. 
It is unknown whether the coherence measure induced by trace norm can be applied for general quantum states. 
Our findings complement the results of coherence quantification in Ref.~\cite{Baumgratz13}.

\begin{acknowledgments}
Z.J. Xi is supported by NSFC (61303009), and the Higher School Doctoral Subject Foundation of Ministry of Education of China (20130202120002). Y.M. Li is supported by NSFC (11271237). H. Fan is supported by 973
program (2010CB922904), NSFC (11175248).
\end{acknowledgments}

\end{document}